\documentclass[letterpaper, 10 pt, conference]{ieeeconf}  

\IEEEoverridecommandlockouts                              
\IEEEoverridecommandlockouts                              
\usepackage{times} 
\usepackage{amsmath} 
\usepackage{amssymb}  
\usepackage{graphicx}  
\usepackage{array}
\usepackage{times}
\usepackage{comment}
\usepackage{theorem}
\usepackage{amsmath,amsfonts,amssymb,amsbsy}

\newcommand{\tr}{\top}
\newcommand{\Nset}{\mathbb{N}}
\newcommand{\Rset}{\mathbb{R}}
\newcommand{\Uset}{\mathbb{U}}

\newcommand{\Fset}{\mathbb{F}}

\newcommand{\Yset}{\mathbb{Y}}

\newcommand{\cS}{\mathcal{S}}

\newcommand{\cE}{\mathcal{E}}
\newcommand{\cN}{\mathcal{N}}

\newcommand{\tini}{\text{ini}}
\newcommand{\bu}{\mathbf{u}}

\newcommand{\by}{\mathbf{y}}
\newcommand{\bg}{\mathbf{g}}
\newcommand{\bY}{\mathbf{Y}}
\newcommand{\bU}{\mathbf{U}}

\usepackage{graphicx}
\usepackage{svg}
\usepackage{subcaption}
\usepackage{multirow}

\usepackage{array}
\newcolumntype{P}[1]{>{\centering\arraybackslash}p{#1}}
\newcolumntype{M}[1]{>{\centering\arraybackslash}m{#1}}

\newcommand{\col}{\operatorname{col}}%

\usepackage{epsfig} 
\usepackage{amsmath} 
\usepackage{amssymb}  

{\theorembodyfont{\slshape}\newtheorem{theorem}{Theorem}[section]}
{\theorembodyfont{\slshape}}
{\theorembodyfont{\slshape}\newtheorem{lemma}[theorem]{Lemma}}
{\theorembodyfont{\slshape}}
{\theorembodyfont{\slshape}\newtheorem{corollary}[theorem]{Corollary}}
{\theorembodyfont{\upshape}\newtheorem{definition}[theorem]{Definition}}
{\theorembodyfont{\upshape}\newtheorem{problem}[theorem]{Problem}}
{\theorembodyfont{\upshape}\newtheorem{remark}[theorem]{Remark}}
{\theorembodyfont{\upshape}}

\title{\LARGE \bf
Basis functions nonlinear data-enabled predictive control: \\ Consistent and computationally efficient formulations  
}

\author{M. Lazar
\thanks{The author is with the Department of Electrical Engineering, Eindhoven University of Technology, The Netherlands, E-mail:
        {\tt\small m.lazar@tue.nl}.}%
}

\begin{document}

\maketitle
\thispagestyle{empty}
\pagestyle{empty}

\begin{abstract}
This paper considers the extension of data-enabled predictive control (DeePC) to nonlinear systems via general basis functions. Firstly, we formulate a basis functions DeePC behavioral predictor and we identify necessary and sufficient conditions for equivalence with a corresponding basis functions multi-step identified predictor. The derived conditions yield a dynamic regularization cost function that enables a well-posed (i.e., consistent) basis functions formulation of nonlinear DeePC. To optimize computational efficiency of basis functions DeePC we further develop two alternative formulations that use a simpler, sparse regularization cost function and ridge regression, respectively. Consistency implications for Koopman DeePC as well as several methods for constructing the basis functions representation are also indicated. The effectiveness of the developed consistent basis functions DeePC formulations is illustrated on a benchmark nonlinear pendulum state-space model, for both noise free and noisy data.
\end{abstract}

\section{Introduction}
\label{sec1}
Model predictive control (MPC) \cite{Rawlings2017} is one of the most successful advanced control methodologies due to its capability to handle constraints, optimize performance and anticipate reference changes. A prediction model can be obtained from physics or using identification. When such a model is simulated to predict future trajectories, modeling/identification errors propagate. To mitigate this, a subspace predictive control (SPC) algorithm was developed in \cite{FavoreelSPC1999} using subspace identification techniques, which yields an unbiased identified multi-step predictor and removes the need of a state estimator. More recently, a data-enabled predictive control (DeePC) algorithm was developed in \cite{CoulsonDeePC2019}, along with closed-loop stability guarantees provided in \cite{Berberich_2021}, which utilizes a behavioral, data-based multi-step predictor. DeePC provides most freedom for optimizing the system response as it jointly solves an identification, estimation and control problem, but has higher computational complexity in general. 

For noise free data, equivalence of MPC and DeePC was established in \cite{CoulsonDeePC2019} and equivalence of SPC and DeePC was established in \cite{Fiedler_2021}. In the case of noisy data, DeePC requires a regularization cost function to yield consistent predictions \cite{FlorianLS}. With the theoretical foundation of DeePC for linear systems sufficiently developed, extensions of DeePC to nonlinear systems were pursued, as most real-life systems are nonlinear. This has proven to be a difficult problem, due to the fact that the  DeePC formulation relies on Willem's fundamental lemma \cite{Willems_2005} for linear systems. 

One of the first approaches that implemented DeePC for a nonlinear drone model \cite{Elokda_2021} directly utilized a linear DeePC formulation along with a regularization cost that penalized the deviation from a known reference input and output trajectory. In \cite{LPV_DPC:Verhoek2021}, a fundamental lemma for linear parameter varying (LPV) systems was derived, which enabled a LPV formulation of DeePC. This offered an effective way to design DeePC algorithms for nonlinear systems that can be represented in an input-output LPV form. A fundamental lemma extension for systems linear in the state, but possibly with input and/or output nonlinear terms was derived in \cite{BerberichNL2020}, which also suggested using basis functions for learning the nonlinear terms. In \cite{Berberich_2022II}, a data-driven predictive control algorithm for nonlinear systems was developed based on data-driven linearization along closed-loop trajectories. In \cite{alsalti2023datadriven} data-driven predictive control methods were developed for feedback linearizable nonlinear systems. 

In \cite{lian_knonline_2021}, reproducing kernel functions were utilized to obtain a linear lifted state-space representation of a nonlinear state-space model, accompanied with a corresponding DeePC formulation. In \cite{lian2021koopman}, a similar path was followed but by using a linear in control input Koopman lifting \cite{Korda_2018_Koop} instead. Both these papers provide valuable insights into nonlinear DeePC design, but arrive at more complex, bilevel optimization problems. These problems arise due to the fact that the identification part is decoupled from the control synthesis part, in an attempt to arrive at consistent predictions. Under some assumptions, the bilevel problem can be reduced to a standard optimization problem via optimality conditions and simulations in \cite{lian2021koopman} show that the corresponding predicted trajectories are close to Koopman model based trajectories.  

In \cite{huang2022robustK}, the reproducing kernel Hilbert space (RKHS) theory was utilized to construct a linear-in-the-parameters multi-step predictor of the nonlinear autoregressive exogenous (NARX) type  \cite{Billings2013}. Then a kernelized DeePC formulation was developed therein based on the analytic closed-form kernelized NARX predictor, along with methods for robust predictions in the presence of noisy output measurements. The approach of \cite{huang2022robustK} leads to a standard, nonlinear programming problem to be solved online, or even to a single quadratic program (QP) for linear in  control input kernel functions.   

In \cite{fazzi2023datadriven}, a structured basis functions representation of one-step NARX models was used to derive a DeePC like algorithm, called \emph{tracking problem} therein. The basis function representation is separated/structured per past inputs/outputs and future inputs/outputs. The authors then developed a sequential algorithm for solving the data-driven tracking problem, which is based on simulating the basis functions one-step NARX model. 

In this paper we consider the formulation of DeePC based on general basis functions transformation of system trajectories. To analyse consistency of the basis function DeePC predictions, we construct a corresponding basis functions subspace predictive control (SPC) problem, which utilizes unbiased multi-step NARX predictors. Firstly, we provide necessary and sufficient conditions under which the predictions of basis functions DeePC are consistent (i.e., in a model equivalence sense \cite{Zadeh_62, Eykhoff}) with the predictions of basis functions SPC. Based on these conditions, we then derive a dynamic regularization cost function that enforces consistency of the basis functions DeePC predictions. To optimize computational efficiency, we further develop two alternative formulations of basis functions DeePC, using a simpler, sparse regularization cost and ridge regression \cite{RegLS_2022}, respectively. 

\section{Preliminaries}
\label{sec2}
Throughout this paper, for any finite number $q\in\Nset_{\geq 1}$ of vectors or functions $\{\xi_1,\ldots,\xi_q\}$ we will make use of the operator $\col(\xi_1,\ldots,\xi_q):=[\xi_1^\tr,\ldots,\xi_q^\tr]^\tr$.

We consider MIMO nonlinear systems with inputs $u \in \Rset^m$ and measured outputs $y \in \Rset^p$ that can be represented using an unknown controllable and observable discrete-time state-space model:
\begin{equation}
 	\label{eq:2.1}
 	\begin{split}
 		x(k+1)&=\tilde f(x(k),u(k)), \quad k\in\Nset,\\
 		y(k)&=\tilde h(x(k)),\end{split}
 \end{equation}
where $x\in\Rset^n$ is an unknown state and $\tilde f$, $\tilde h$ are suitable functions. In model predictive control \cite{Rawlings2017}, given an initial measured or estimated state at time $k\in\Nset$, the above system equations are used to compute a sequence of predicted outputs $\by_{[1,N]}(k):=\{y(1|k),\ldots,y(N|k)\}$, given a sequence of predicted inputs $\bu_{[0,N-1]}(k):=\{u(0|k),\ldots,u(N-1|k)\}$.

In \emph{indirect} data-driven predictive control, see, e.g., \cite{masti2020learning}, typically multi-step predictors of the NARX type \cite{Billings2013} are identified from the input/measured output data generated by system \eqref{eq:2.1} and used to predict $\by_{[1,N]}(k)$ from past inputs and outputs and $\bu_{[0,N-1]}(k)$, i.e.,
\begin{equation}
	\label{eq:2.2}
	\by_{[1,N]}(k):=\mathbb{F}(\bu_\tini(k),\by_\tini(k),\bu_{[0,N-1]}(k)),
\end{equation} 
where $\mathbb{F}:=\col(f_1,\ldots,f_N)$ and
\begin{align*}
	\bu_\tini(k)&:=\col(u(k-T_\tini),\ldots,u(k-1)) \in \Rset^{T_\tini m},\\ 
	\by_\tini(k)&:=\col(y(k-T_\tini+1),\ldots,y(k))\in \Rset^{
 T_\tini p},
\end{align*}
and where $T_\tini \in \mathbb{N}_{\geq 1}$ defines the order of the NARX dynamics (different orders can be used for inputs and outputs, but for simplicity we use a common order). Note that since each $f_i$ is a MIMO predictor, it is the aggregation of several MISO predictors, i.e., $f_i=\col(f_{i,1},\ldots,f_{i,p})$ where each $f_{i,j}$ predicts the $j$-th output, i.e., for $i=1,\ldots,N$
\begin{align*}
	y_j(i|k)&=f_{i,j}(\bu_\tini(k),\by_\tini(k),\bu_{[0,N-1]}(k)),\\
	y(i|k)&=\col(y_1(i|k),\ldots,y_p(i|k)),
\end{align*} 
where $j=1,\ldots,p$ and $p$ is the number of outputs. As opposed to a standard, one-step NARX prediction model, which can produce a $N$-step ahead prediction of the output by simulation, the \emph{multi-step} NARX predictors \eqref{eq:2.2} directly compute a $N$-step ahead prediction of the output by evaluating in parallel $N$ functions $\{f_1,\ldots,f_N\}$ that share the same arguments, i.e., $\{\bu_\tini(k),\by_\tini(k),\bu_{[0,N-1]}(k)\}$. 

Next, for any $k\geq 0$ (starting time instant in the data vector) and  $j\geq 1$ (length of the data vector obtained from system \eqref{eq:2.1}), define
\begin{align*}
	\bar\bu(k,j) &:= \col(u(k),\ldots, u(k+j-1)),\\
    \bar\by(k,j)&: = \col(y(k), \ldots, y(k+j-1)).
\end{align*}
Then we can define the Hankel matrices:
\begin{equation}\label{eq:hankel_data}
	\begin{aligned}
		\bU_p &:= \begin{bmatrix}\bar\bu(0,T_\tini) & \ldots & \bar\bu(T-1, T_\tini) \end{bmatrix}, \\
		\bY_p &:= \begin{bmatrix}\bar\by(1, T_\tini) & \ldots & \bar\by(T, T_\tini) \end{bmatrix},\\
		\bU_f &:= \begin{bmatrix}\bar\bu(T_\tini, N) & \ldots & \bar\bu(T_\tini+T-1, N) \end{bmatrix},\\
		\bY_f &:= \begin{bmatrix}\bar\by(T_\tini+1, N) & \ldots & \bar\by(T_\tini+T, N) \end{bmatrix},
	\end{aligned}
\end{equation}
where $T\geq (m+p)T_\tini+mN$ is the number of columns of the Hankel matrices. Note that we consider Hankel matrices for simplicity of exposition; the results developed in this paper hold for general suitable data matrices, such as page matrices. 

Then, if we parameterize the multi-step subspace predictor $\mathbf{F}(\bu_\tini,\by_\tini,\bu_{0,N-1},\Theta)$ using a matrix of parameters $\Theta$ of suitable dimensions, we can formulate the nonlinear least squares problem:
\begin{equation}
	\label{eq:2.3}
	\min_{\Theta}\left\|\bY_f-\Fset(\bU_p,\bY_p,\bU_f,\Theta)\right\|_F^2,
\end{equation}
where $\|\cdot\|_F$ denotes the matrix Frobenius norm and, with a slight abuse of notation, $\Fset(\bU_p,\bY_p,\bU_f,\Theta)$ denotes applying the operator $\Fset$ in \eqref{eq:2.2} to every column in the matrix $\begin{bmatrix}\bU_p\\\bY_p\\ \bU_f \end{bmatrix}$. 

Most of the existing \emph{direct} (i.e., behavioral multi-step predictors) and \emph{indirect} (identified multi-step predictors) approaches to nonlinear data-driven predictive control can be related to such a \emph{parameterized} multi-step NARX predictor. In the next section we will introduce a general basis functions representation of $\Fset$ and we will analyze the resulting indirect and direct data-driven predictive controllers, respectively.  

\section{Main results}
\label{sec3}
In what follows we will define a basis functions representation of the multi-step nonlinear predictor \eqref{eq:2.2}. To this end, for every predicted output $y_j$, $j=1,\ldots,p$, consider a finite set of basis functions $\{\phi_0,\phi_1\ldots,\phi_L\}$, $L\in\Nset_{\geq 1}$ and define
\begin{align*}
	y_j(i|k)&=\sum_{l=0}^L \theta_{l,j}^i \phi_l(\bu_\tini(k),\by_\tini(k),\bu_{[0,N-1]}(k))\\
	&=\begin{bmatrix}\theta_{0,j}^i&\ldots&\theta_{L,j}^i\end{bmatrix}\bar\phi(\bu_\tini(k),\by_\tini(k),\bu_{[0,N-1]}(k)),
\end{align*} 
which corresponds to a basis functions representation of the MISO functions $f_{i,j}$ using a common set of basis functions. Above $\bar \phi:=\col(\phi_0,\ldots,\phi_L)$ and $\phi_0(\cdot):=1$ to account for a vector of affine constant terms (biases) in \eqref{eq:2.2}. By stacking up all predicted outputs for all future time instants $i=1,\ldots,N$ we obtain a linear-in-the-parameters representation of the NARX multistep predictor $\Fset$, i.e.,
\begin{equation}
	\label{eq:3.1}
	\by_{[1,N]}(k):=\Theta\bar\phi(\bu_\tini(k),\by_\tini(k),\bu_{[0,N-1]}(k)).
\end{equation}   
Next, define 
\begin{equation}
\label{eq:3.phi}
\Phi:=\bar\phi(\bU_p,\bY_p,\bU_f)\in\Rset^{(L+1)\times T}
\end{equation}
as a matrix of data obtained by applying the operator $\bar\phi$ to every column in the matrix of data  $\begin{bmatrix}\bU_p\\\bY_p\\ \bU_f \end{bmatrix}$. More precisely, the element in line $i$, $1\leq i\leq L+1$ and column $j$, $1\leq j \leq T$ of $\Phi$ is given by  $\phi_{i-1}\left(\begin{bmatrix}\bU_p\\\bY_p\\ \bU_f \end{bmatrix}_{:j}\right)$, where $Q_{:j}$ denotes the $j$-th column of any matrix $Q$. Since every column of the data matrix is a system trajectory, $\Phi$ represents the basis functions transformation of these trajectories. 

Then the nonlinear least squares problem \eqref{eq:2.3} becomes the least squares problem:
\begin{equation}
	\label{eq:3.2}
	\min_{\Theta}\left\|\bY_f-\Theta\Phi\right\|_F^2.
\end{equation}
Assuming that the input data and the set of basis functions are such that $\Phi$ has full row-rank, which corresponds to persistency of excitation in the space of transformed trajectories, we obtain the least squares optimal solution:
\begin{equation}
	\label{eq:3.theta}
	\Theta^\ast:=\bY_f\Phi^\dagger=\bY_f\Phi^\top(\Phi\Phi^\top)^{-1}.
\end{equation}
Next we can define the basis functions SPC controller. 
\begin{problem}[$\phi$-SPC]
\label{prob:SPC}
\begin{subequations}
	\label{eq:3.SPC}
	\begin{align}
		\min_{\Xi(k)}\quad  & l_N(y(N|k))+\sum_{i=0}^{N-1} l_s(y(i|k), u(i|k))\label{eq:3.SPCa}\\
		&\text{subject to constraints:}\nonumber\\
		&\by_{[1,N]}(k)=\Theta^\ast\bar\phi(\bu_\tini(k),\by_\tini(k),\bu_{[0,N-1]}(k))\label{eq:3.SPCb}\\
		&(\by_{[1,N]}(k),\bu_{[0,N-1]}(k))\in\Yset^N\times\Uset^N\label{eq:3.SPCc}
	\end{align}
\end{subequations}
\end{problem}
where $\Xi(k):=\col(\by_{[1,N]}(k),\bu_{[0,N-1]}(k))$ are the optimization variables, $\Yset, \Uset$ are proper polytopic sets that represent constraints, $l_s(y,u):=\|y-y_r\|_Q^2+\|u-u_r\|_R^2$ is a stage cost and $l_N(y)$ is a terminal cost, taken for simplicity as $l_s(y,0)$.

The basis functions DeePC controller is defined next.
\begin{problem}[$\phi$-DeePC]
\label{prob:DeePC}
\begin{subequations}
	\label{eq:3.DeePC}
	\begin{align}
		\min_{\Xi(k)}\quad  & l_N(y(N|k))+\sum_{i=0}^{N-1} l_s(y(i|k), u(i|k))\label{eq:3.DeePCa}\\
		&\text{subject to constraints:}\nonumber\\
		& \begin{bmatrix}\Phi \\ \bY_f\end{bmatrix}\bg(k)=\begin{bmatrix}\bar\phi(\bu_\tini(k),\by_\tini(k),\bu_{[0,N-1]}(k))\\\by_{[1,N]}(k)\end{bmatrix}\label{eq:3.DeePCb}\\
		&(\by_{[1,N]}(k),\bu_{[0,N-1]}(k))\in\Yset^N\times\Uset^N\label{eq:3.DeePCc}
	\end{align}
\end{subequations}
\end{problem}
where $\Xi(k):=\col(\by_{[1,N]}(k),\bu_{[0,N-1]}(k),\bg(k))$ are the optimization variables. 

Next, to analyze the relation between the $\phi$-SPC predictions and the $\phi$-DeePC predictions, we introduce the following notion of system model or system representation equivalence inspired by \cite{Zadeh_62, Eykhoff}.
\begin{definition}
\label{def:eykoff}
Two models $\{M_1, M_2\}$ of system \eqref{eq:2.1} are called equivalent if for every constraints admissible input sequence $\bu_{[0,N-1]}$ and initial condition it holds that
\begin{equation}
\label{eq:echiv}
\|\by_{[1,N]}^{M_1}-\by_{[1,N]}\|=\|\by_{[1,N]}^{M_2}-\by_{[1,N]}\|,
\end{equation}
where $\by_{[1,N]}$ is the true system \eqref{eq:2.1} output. 
\end{definition}
Notice that one way to establish model equivalence is to show that $\by_{[1,N]}^{M_1}=\by_{[1,N]}^{M_2}$. The above model equivalence notion is very useful for nonlinear data-driven predictive control, as in the nolinear case, the Willems' fundamental lemma \cite{Willems_2005} does not hold. Indeed, given a trustworthy, unbiased identified model, such as \eqref{eq:3.1}, we can use its predictions as a guiding standard for achieving consistent predictions in nonlinear ($\phi$-)DeePC. \emph{Hence, for want of a nonlinear fundamental lemma all is not lost, as long as equivalence with a consistent model is guaranteed, as stated next.}
\begin{lemma}
\label{lem:eq}
Consider the $\phi$-SPC prediction model \eqref{eq:3.SPCb} and the $\phi$-DeePC prediction model \eqref{eq:3.DeePCb} defined using the same set of data $\{\bU_p,\bY_p,\bU_f,\bY_f\}$ generated using system \eqref{eq:2.1} and the same set of basis functions $\{\phi_0,\phi_1,\ldots\phi_L\}$. Assume that the basis functions transformed data matrix $\Phi$ has full row-rank. Let $\cE:=\bY_f-\Theta^\ast\Phi$ be the matrix of residuals of the least squares problem \eqref{eq:3.2} and let $\cS_\bg:=\{\Phi^\dagger\bar\phi(\cdot)+\hat \bg\ : \ \hat \bg\in\cN(\Phi)\}\subset\Rset^{T}$ be a set of parameters $\bg$, where $\cN(\Phi)$ is the null-space of $\Phi$. 
Then the $\phi$-SPC prediction model \eqref{eq:3.SPCb} is equivalent with the $\phi$-DeePC prediction model \eqref{eq:3.DeePCb} \emph{if and only if}  $\cE\hat\bg=\mathbf{0}$ for all $\hat \bg\in\cN(\Phi)$.
\end{lemma}
\begin{proof}
The proof follows a similar reasoning as in the proof of Theorem~1 in \cite{Fiedler_2021}, \emph{mutatis mutandis}. From \eqref{eq:3.DeePCb}, it follows that 
\[\Phi\bg(k)=\bar\phi(\bu_\tini(k),\by_\tini(k),\bu_{[0,N-1]}(k))\]
and thus all variables $\bg(k)$ that satisfy this system of equations satisfy
\[\bg(k)\in\cS_\bg=\{\Phi^\dagger\bar\phi(\cdot)+\hat \bg\ : \ \hat \bg\in\cN(\Phi)\}.\]
Therefore, all predicted outputs generated by $\phi$-DeePC satisfy
\[\by_{[1,N]}(k)\in\{\bY_f\Phi^\dagger\bar\phi(\cdot)+\bY_f\hat \bg\ : \ \hat \bg\in\cN(\Phi)\}.\]
Based on \eqref{eq:3.2}-\eqref{eq:3.theta} and $\hat \bg\in\cN(\Phi)$, we have that 
\begin{align*}
\bY_f\hat \bg&=\left(\cE+\Theta^\ast\Phi\right)\hat \bg=\cE\hat \bg+\Theta^\ast\Phi\hat \bg\\
&=\cE\hat \bg.
\end{align*}
Thus, it holds that
\[\by_{[1,N]}^{\phi-\text{DeePC}}(k)=\bY_f\Phi^\dagger\bar\phi(\cdot)=\Theta^\ast\bar\phi(\cdot)=\by_{[1,N]}^{\phi-\text{SPC}}\]
\emph{if and only if} $\cE\hat\bg=\mathbf{0}$ for all $\hat \bg\in\cN(\Phi)$, which completes the proof.
\end{proof}
\begin{remark}
In the deterministic, noise free, linear case, the above result recovers the result of Theorem~1 in \cite{Fiedler_2021} because in the linear case the outputs can be exactly predicted from finite persistently exciting data, i.e., $\cE=\mathbf{0}$. However, the proof of Lemma~\ref{lem:eq} shows that in the case of noise free data generated by a nonlinear system (or in the case of \emph{noisy data} generated by a linear or nonlinear system) the $\phi$-DeePC predictor will not necessarily be consistent with the unbiased $\phi$-SPC predictor. Hence, for nonlinear systems, even in the deterministic, noise free, case a regularization cost is required to enforce consistent predictions. Alternatively, the basis functions should be such that the outputs of the nonlinear system can be exactly predicted from finite data, which is typically hard if not impossible to achieve for general nonlinear systems when using a finite set of basis functions.  
\end{remark}

Hence, since in general in the deterministic nonlinear case or the noisy data case $\cE\neq\mathbf{0}$, the remaining option to achieve consistent predictions in $\phi$-DeePC is to regularize the variables $\bg$ such that 
$\cS_\bg\approx\{\Phi^\dagger\bar\phi(\cdot)\}$ or, alternatively, to regularize the variables $\hat\bg$ to zero. The first option yields the following regularized basis function DeePC controller.
\begin{problem}[$\phi$-DeePC-R1]
\label{prob:DeePCR1}
\begin{subequations}
	\label{eq:3.DeePCR1}
	\begin{align}
		\min_{\Xi(k)}\quad  & l_N(y(N|k))+\sum_{i=0}^{N-1} l_s(y(i|k), u(i|k))+l_g^1(\bg(k))\label{eq:3.DeePCR1a}\\
		&\text{subject to constraints:}\nonumber\\
		& \begin{bmatrix}\Phi \\ \bY_f\end{bmatrix}\bg(k)=\begin{bmatrix}\bar\phi(\bu_\tini(k),\by_\tini(k),\bu_{[0,N-1]}(k))\\\by_{[1,N]}(k)\end{bmatrix}\label{eq:3.DeePCR1b}\\
		&(\by_{[1,N]}(k),\bu_{[0,N-1]}(k))\in\Yset^N\times\Uset^N\label{eq:3.DeePCR1c}
	\end{align}
\end{subequations}
\end{problem}
where 
\begin{equation}
\label{eq:3.lgR1}
l_g^1(\bg(k)):=\lambda\|\bg(k)-\bg_r(k)\|_2^2,
\end{equation}
and $\bg_r(k):=\Phi^\dagger\bar\phi(\bu_\tini(k),\by_\tini(k),\bu_{[0,N-1]}(k))$.

The second option indicated above leads to the following alternative regularized $\phi$-DeePC controller. 

\begin{problem}[$\phi$-DeePC-R2]
\label{prob:DeePCR2}
\begin{subequations}
	\label{eq:3.DeePCR2}
	\begin{align}
		&\min_{\Xi(k)}\quad   l_N(y(N|k))+\sum_{i=0}^{N-1} l_s(y(i|k), u(i|k))+l_g^2(\hat \bg(k))\label{eq:3.DeePCR2a}\\
		&\text{subject to constraints:}\nonumber\\
		& \Phi\hat\bg(k)=0\label{eq:3.DeePCR2b}\\
  &\bY_f\left(\Phi^\dagger\bar\phi(\bu_\tini(k),\by_\tini(k),\bu_{[0,N-1]}(k))+\hat\bg(k)\right)=\by_{[1,N]}(k)\label{eq:3.DeePCR2c}\\
		&(\by_{[1,N]}(k),\bu_{[0,N-1]}(k))\in\Yset^N\times\Uset^N\label{eq:3.DeePCR2d}
	\end{align}
\end{subequations}
\end{problem}
where 
\begin{equation}
\label{eq:3.lgR2}
l_g^2(\hat \bg(k)):=\lambda\|\hat\bg(k)\|_2^2
\end{equation}
and $\Xi(k):=\col(\by_{[1,N]}(k),\bu_{[0,N-1]}(k),\hat \bg(k))$ are the optimization variables.
From the proof of Lemma~\ref{lem:eq} it directly follows that $\by_{[1,N]}^{\phi-\text{DeePC}-R1}(k)=\by_{[1,N]}^{\phi-\text{DeePC}-R2}(k)$ for the same cost function weighting matrices and $\lambda$ parameter, i.e., since $\bg(k)-\bg_r(k)=\hat\bg(k)$ for some $\hat\bg(k)\in\cN(\Phi)$. However, Problem~\ref{prob:DeePCR1} uses a dynamic nonlinear cost function, while Problem~\ref{prob:DeePCR2} uses a static quadratic cost function with a sparse Hessian. Thus, the $\phi$-DeePC-R2 formulation is preferable and it will likely yield smaller computation times. 
\begin{remark}
The result of Lemma~\ref{lem:eq} only assumes that the matrix $\Phi$ has full row-rank, i.e., it does not assume noise free data. Hence, the $\phi$-SPC equivalence conditions of Lemma~\ref{lem:eq} also hold for noisy output data. Moreover, the $\phi$-DeePC-R2 predictor equations \eqref{eq:3.DeePCR2b}-\eqref{eq:3.DeePCR2c} always admit the solution $\hat\bg(k)=\mathbf{0}$ even when the output measurements are affected by noise, case in which the $\phi$-DeePC-R2 predictor reduces to the unbiased least squares optimal $\phi$-SPC predictor. Hence, the regularized $\phi$-DeePC-R2 formulation developed in this paper enforces consistency in a least squares sense with respect to both modeling errors and measurement noise. The $\phi$-DeePC-R1 formulation on the other hand is more suitable for tuning the parameter $\lambda$ via the Hanke-Raus heuristic \cite{Hanke96}. Since the two regularized formulations are equivalent, the $\lambda$ tuned for cost \eqref{eq:3.DeePCR1a} can then be used in the cost \eqref{eq:3.DeePCR2a}.
\end{remark}

\subsection{Ridge regression $\phi$-DeePC formulation}
In practice it may be difficult to arrive at a set of basis functions and input-output data such that the matrix $\Phi$ has full row-rank, which relates to persistently exciting data generation for nonlinear systems. In this case, a ridge regression solution \cite{RegLS_2022} to the least squares problem \eqref{eq:3.2} can be computed as 
\begin{equation}
	\label{eq:3.thetar}
\Theta^{R\ast}:=\bY_f\Phi^\top(\Phi\Phi^\top+\gamma I)^{-1}
\end{equation}
where $\gamma$ is a positive scalar and $I$ is an identity matrix of suitable dimensions. Then we can formulate a corresponding basis functions DeePC controller as follows.
\begin{problem}[Ridge $\phi$-DeePC]
\label{prob:DeePCR}
\begin{subequations}
	\label{eq:3.DeePCR}
	\begin{align}
		\min_{\Xi(k)}&\quad  l_N(y(N|k))+\sum_{i=0}^{N-1} l_s(y(i|k), u(i|k))\label{eq:3.DeePCRa}\\
		&\text{subject to constraints:}\nonumber\\
		& \begin{bmatrix}\Phi\Phi^\top+\gamma I \\ \bY_f\Phi^\top\end{bmatrix}\bg(k)=\begin{bmatrix}\bar\phi(\bu_\tini(k),\by_\tini(k),\bu_{[0,N-1]}(k))\\\by_{[1,N]}(k)\end{bmatrix}\label{eq:3.DeePCRb}\\
		&(\by_{[1,N]}(k),\bu_{[0,N-1]}(k))\in\Yset^N\times\Uset^N.\label{eq:3.DeePCRc}
	\end{align}
\end{subequations}
\end{problem}
\begin{lemma}
\label{lem:eqR}
Consider the $\phi$-SPC prediction model \eqref{eq:3.SPCb} with $\Theta^\ast$ replaced by $\Theta^{R\ast}$ as in \eqref{eq:3.thetar} and the $\phi$-DeePC prediction model \eqref{eq:3.DeePCRb} defined using the same set of data $\{\bU_p,\bY_p,\bU_f,\bY_f\}$ generated using system \eqref{eq:2.1} and the same set of basis functions $\{\phi_0,\phi_1,\ldots\phi_L\}$. Assume that $\gamma>0$ is such that $\Phi\Phi^\top+\gamma I\succ 0$.

Then the ridge regression $\phi$-SPC prediction model \eqref{eq:3.SPCb} is equivalent with the ridge regression $\phi$-DeePC prediction model \eqref{eq:3.DeePCRb}.
\end{lemma}
\begin{proof}
The proof simply follows by observing that the ridge $\phi$-DeePC predictor is now uniquely defined as 
\[
\begin{split}
\by_{[1,N]}(k)&=\bY_f\Phi^\top\bg(k)=\bY_f\Phi^\top(\Phi\Phi^\top+\gamma I)^{-1}\bar\phi(\cdot)\\
&=\Theta^{R\ast}\bar\phi(\cdot).
\end{split}
\]
\end{proof}
\begin{remark}
The ridge $\phi$-DeePC formulation offers less flexibility to optimize the bias variance trade-off compared to the $\phi$-DeePC-R2 formulation, but it can handle $\Phi$ matrices without a full row-rank and it offers more flexibility to reduce computational complexity. Indeed, notice that in ridge $\phi$-DeePC the dimension of the vector of variables $\bg(k)\in\Rset^{L+1}$ is dictated by the number of basis functions $L+1$ versus the data size $T$, as for $\phi$-DeePC(-R1,-R2). This makes a huge difference in the nonlinear case, where typically a large data size is required; e.g., in the illustrative example in Section~\ref{sec4}, $T=990$ and $L+1=40$ yields good performance.
\end{remark}
\begin{remark}
The kernelized DeePC formulation based on the RKHS theory presented in \cite{huang2022robustK} generates the matrix $\Phi$ as a so-called Gram matrix, which is always square and positive semi-definite. This can be regarded as a specific way of generating the basis functions in ridge $\phi$-DeePC, where the number of basis functions is $L+1=T$ and the predictor \eqref{eq:3.DeePCRb} becomes:
\[\begin{bmatrix}\Phi+\gamma I \\ \bY_f\end{bmatrix}\bg(k)=\begin{bmatrix}\bar\phi(\bu_\tini(k),\by_\tini(k),\bu_{[0,N-1]}(k))\\\by_{[1,N]}(k)\end{bmatrix}.\]
The advantage in this case is that due to the RKHS theory, $\Phi\succeq 0$ holds for any system trajectory; the disadvantage is however that $\bg(k)\in\Rset^T$ and the computational complexity is again dictated by the data size $T$. 
\end{remark}
\subsection{Relation with Koopman MPC}
Koopman model predictive control was developed in \cite{Korda_2018_Koop} as a method to apply linear MPC techniques to nonlinear systems. To this end, the idea is to lift the state $x(k)$ of the original system \eqref{eq:2.1} to a higher dimensional space where the dynamics are linear, via a set of observables, which can be parameterized using basis functions, i.e.,
\[z(k):=\bar\phi_K(x(k)):=\col(\phi_{1,K}(x(k)),\ldots,\phi_{L,K}(x(k))).\]
This can also be done based on input-output data as presented in \cite{Korda2020}, which yields the following linear in control input embedding of the nonlinear system \eqref{eq:2.1}:
\begin{equation}
\label{eq:3.sysK}
\begin{split}
z(k+1)&=Az(k)+Bu(k),\quad k\in\Nset,\\
y(k)&=Cz(k),\\
z(0)&=\bar\phi_K(\bu_\tini(0),\by_\tini(0)).
\end{split}
\end{equation}
The above model can be used to define $N$-step ahead MPC prediction matrices, i.e.,
\[
	\begin{aligned}
		\Psi & := \begin{bmatrix}CA \\ CA^2 \\ \vdots \\ CA^N\end{bmatrix}, ~~ 
		\Gamma := \begin{bmatrix} CB & 0 & \dddot{} & 0 \\ CAB & CB & \dddot{} & 0 \\ \vdots & \vdots & \ddots & \vdots \\  CA^{N-1}B & CA^{N-2}B & \dddot{} & CB \end{bmatrix}, 
	\end{aligned}
\]
which yields the following Koopman MPC controller.
\begin{problem}[Koopman MPC]
\label{prob:Koop}
\begin{subequations}
	\label{eq:3.Koop}
	\begin{align}
		\min_{\Xi(k)}\quad  & l_N(y(N|k))+\sum_{i=0}^{N-1} l_s(y(i|k), u(i|k))\label{eq:3.Koopa}\\
		&\text{subject to constraints:}\nonumber\\
		&\by_{[1,N]}(k)=\Psi z(0|k)+\Gamma\bu_{[0,N-1]}(k)\label{eq:3.Koopb}\\
  &z(0|k)=\bar\phi_K(\bu_\tini(k),\by_\tini(k))\label{eq:3.Koopc}\\
		&(\by_{[1,N]}(k),\bu_{[0,N-1](k)}(k))\in\Yset^N\times\Uset^N,\label{eq:3.Koopd}
	\end{align}
\end{subequations}
\end{problem}
where $\Xi(k):=\col(\by_{[1,N]}(k),\bu_{[0,N-1]}(k))$ are the optimization variables. Then, since the lifted model \eqref{eq:3.sysK} is linear, the following result is a consequence of the results in \cite{CoulsonDeePC2019}, \cite{Fiedler_2021}.
\begin{corollary}
 \label{cor:eqK}
Assume that the Koopman lifted state-space system \eqref{eq:3.sysK} is controllable and observable. Assume that a persistently exciting input is used to generate noise free output data for system \eqref{eq:3.sysK} such that the matrix 
$\Phi_K:=\bar\phi(\bU_p,\bY_p,\bU_f)$
has full row-rank, where 
\begin{equation}
\label{eq:3.Kbasis}
\bar\phi(\bu_\tini,\by_\tini,\bu_{[0,N-1]}):=\begin{bmatrix}\bar\phi_K(\bu_\tini,\by_\tini)\\ \bu_{[0,N-1]}\end{bmatrix}=\begin{bmatrix}z\\ \bu_{[0,N-1]}\end{bmatrix}.
\end{equation}
Consider the $\phi$-SPC prediction model \eqref{eq:3.SPCb} and the $\phi$-DeePC prediction model \eqref{eq:3.DeePCb} defined using the same set of data $\{\bU_p,\bY_p,\bU_f,\bY_f\}$ and the same set of basis functions $\{\phi_{1,K},\ldots\phi_{L,K}\}$ and $\bar\phi(\cdot)$ defined as in \eqref{eq:3.Kbasis}. 

Then the Koopman MPC prediction model \eqref{eq:3.Koopb}-\eqref{eq:3.Koopc}, the $\phi$-SPC prediction model \eqref{eq:3.SPCb} and the $\phi$-DeePC prediction model \eqref{eq:3.DeePCb} are equivalent.
\end{corollary}

The above result shows that Koopman DeePC is a special case of basis functions DeePC, i.e., corresponding to basis functions that are linear in the present and future control inputs. This implies that via the regularized $\phi$-DeePC formulations (R1, R2 or Ridge) developed in this paper, we can obtain consistent Koopman DeePC formulations without separating the identification problem from the prediction/control synthesis problem, as proposed in \cite{lian2021koopman}, while still solving a single QP online due to linearity of the basis functions in present and future inputs.

\subsection{Construction of basis functions representations}
The results in Lemma~\ref{lem:eq} and Lemma~\ref{lem:eqR} that relate basis functions based behavioral multi-step predictors with identified multi-step predictors hold for any suitable basis functions representation. Hence, the consistent $\phi$-DeePC formulations (R1, R2 and Ridge, respectively) developed in this paper can be applied for any basis functions, which opens the door to a wide range of powerful machine learning methods \cite{Suykens} that can be exploited in nonlinear DeePC. This includes the following known types of basis functions.
\paragraph{RKHS} The reproducible kernel Hilbert space approach was already used to design a kernelized DeePC algorithm in \cite{huang2022robustK}. In this case the basis functions $\{\phi_1,...,\phi_L\}$ have the following features: they do not include a bias/affine term; besides the arguments of $\Fset$, they also depend on other hyper-parameters, e.g., the centers of the kernel functions. In the RKHS approach the centers are chosen as all the data points, i.e., each column of the matrix $\begin{bmatrix}\bU_p\\\bY_p\\ \bU_f \end{bmatrix}$ is a center. This means that the Gram matrix $\Phi\in\Rset^{T\times T}$ is guaranteed positive semidefinite and square, where $T$ is the number of data points, but the computational complexity is driven by the data size. 

\paragraph{Radial basis functions neural networks (RBF NNs)} Such basis functions $\{\phi_0,...,\phi_L\}$ also depend on a center hyper-parameter, but the number of basis functions/centers can be selected freely, leading to a fat matrix $\Phi\in\Rset^{(L+1)\times T}$; and a bias/affine term is allowed. Compared to RKHS, the centers (and other hyper-parameters) need to be selected/learned. One common approach is to use a $K$-means clustering algorithm to select the centers. A more powerful approach is to use gradient descent methods to learn the centers of RBFs NNs, which can deal with large data sets. 

\paragraph{Orthogonal neural networks (ONNs)} Such neural networks \cite{ONN} have a similar architecture with RBF NNs, but instead use so-called processing units to generate the basis functions. Following the choice of an orthogonal set of functions (e.g., Chebyshev polynomials), each element in the input data $\{\bu_\tini,\by_\tini,\bu_{[0,N-1]}\}$ is passed through all orthogonal functions up to a certain order; then the processing units generate the set of basis functions by taking all possible products over the original set of functions, which are linearly combined to form the output of the ONN. The advantage of ONNs is that by increasing the order of the orthogonal functions there is a guarantee that the approximation error decreases. The limitation is that the products of functions explodes with the number of inputs and the order; also, for most orthogonal basis functions the input and output data must be scaled into bounded intervals before training and then an inverse scaling must be performed to obtain meaningful outputs. 

\paragraph{Taylor series expansion of multivariate functions} As proposed in \cite{masti2020learning}, the operator $\Fset$ can be approximated using a simplified Taylor series, which yields, up to a second order:
\[
\begin{split}
&\mathbb{F}(\bu_\tini,\by_\tini,\bu_{[0,N-1]})\approx\\&\Fset_0(\bu_\tini,\by_\tini)+\Fset_1(\bu_\tini,\by_\tini)(\bu_{[0,N-1]}-\tilde\bu_{[0,N-1]})+\\
& (\bu_{[0,N-1]}-\tilde\bu_{[0,N-1]})^\top \Fset_2(\bu_\tini,\by_\tini) (\bu_{[0,N-1]}-\tilde\bu_{[0,N-1]}).
\end{split}
\]
The expansion vector $\tilde\bu_{[0,N-1]}$ can be taken as zero, or as $u_r$, or as a shifted predicted sequence from the previous time instant. All the elements in the vectors/matrices $\Fset_0, \Fset_1, \Fset_2$ can be represented using any of the basis functions choices suggested above. Then, if a first order Taylor expansion is adopted, one obtains a linear in control inputs predictor which leads to a QP problem for $\phi$-DeePC. If a second order Taylor expansion is used, we obtain quadratically constrained QP problem instead. As shown in \cite{masti2020learning} (therein, general NNs were used to learn $\Fset_0$ and $\Fset_1$), even a first order expansion can produce quite accurate predictions. 

\section{Illustrative example}
\label{sec4}
Consider the following nonlinear pendulum state-space model taken from \cite{Dpn2022}: 
 \begin{equation}\label{eq:sys}
 \begin{split}
\begin{bmatrix}
x_1(k+1) \\
x_2(k+1) 
\end{bmatrix}&=\begin{bmatrix}
1-\frac{b T_s}{J} & 0 \\
T_s& 1
\end{bmatrix}  \begin{bmatrix}
x_1(k) \\
x_2(k) 
\end{bmatrix}+ \begin{bmatrix}
\frac{T_s}{J} \\
0 
\end{bmatrix} u(k)\\
&-\begin{bmatrix}
\frac{MLg T_s}{2J}\sin(x_2(k)) \\
0 
\end{bmatrix} \\
    y(k)&=x_2(k)+w(k),
 \end{split}    
 \end{equation}
where $u(k)$ and $y(k)$ are the system input torque and pendulum angle at time instant $k$, while $J = \frac{ML^2}{3}$, $M= 1 kg$ and $L = 1 m$ are the moment of inertia, mass and length of the pendulum. Moreover, $g = 9.81 m/s^2$ is the gravitational acceleration, $b=0.1$ is the friction coefficient and the sampling time $T_s=\frac{1}{30}$s. The performance of the developed basis function data-enabled and subspace predictive controllers is evaluated for a prediction horizon $N=10$ and tracking cost function weights $Q=200$ and $R=0.5$ for all algorithms. Except for one simulation when multiple $\lambda$ values are specified in Figure~\ref{fig1}, for both $\phi$-DeePC-R1 and -R2 we used $\lambda=1e+4$; for Ridge $\phi$-DeePC we used $\gamma=1e-3$. For the noisy data simulations we used the same white noise realization for all formulations with standard deviation $0.01$, which is 10 times larger than in \cite{Dpn2022}.

To generate the output data an open-loop identification experiment was performed using a multisine input constructed with the Matlab function \emph{idinput}, with the parameters \emph{Range} $[-4,\,4]$, \emph{Band} $[0,\,1]$, \emph{Period} $1000$, \emph{NumPeriod} $1$ and \emph{Sine} $[25,\,40,\,1]$. The data length is $1000$ and $T_\tini=5$ is used, as estimated in \cite{Dpn2022}. For identification/formulation from noisy data white noise was added to the output with standard deviation $0.01$.

We have used basis functions that are linear in the present and future inputs, as defined in \eqref{eq:3.Kbasis}. This allows us to solve all predictive control formulations developed in this paper via QP, using the \emph{quadprog} solver, which provides optimal solutions, so the comparison of the obtained results is not hindered by local minima. To generate the basis functions $\phi_K(\bu_\tini,\by_\tini)$ we utilized a RBF NN representation with 30 centers/neurons with Gaussian activation functions. A hybrid RBF NN was created by adding the 10 linear inputs $\bu_{[0,N-1]}$ to the 30 outputs of the Gaussian neurons in \emph{PyTorch} and then it was trained to find the optimal centers using the \emph{Adam} optimizer with the \emph{MSE} loss function, the \emph{learning rate} $5e-4$ and \emph{$L_2$} regularization using the weight $1e-7$. This yields 30 basis functions $\phi_{K,l}(z):=e^{-\|z-z_{c,l}^\ast\|_2^2}$, where $z=\col(\bu_\tini,\by_\tini)$, $l=1,\ldots,30$. By letting 
\[
\begin{split}
\bar\phi(\bu_\tini,\by_\tini,&\bu_{[0,N-1]})\\&=\col(\phi_{K,1}(z),\ldots,\phi_{K,30}(z), \bu_{[0,N-1]})
\end{split}
\] we obtain a $\Phi\in\Rset^{40\times 990}$ matrix, since $N=10$, the data length is $1000$ and we used  Hankel matrices to generate $(\bU_p,\bY_p,\bU_f)$. The resulting $\Phi$ matrix has full row-rank both for the noise free case and the noisy data case. Once the $\Phi$ matrix is defined, we can directly construct $\phi$-SPC and $\phi$-DeePC(-R1,-R2, Ridge) formulations.

Firstly, we consider the noise free data case and we test the consistency of the $\phi$-DeePC-R1 formulation for a small value of the regularization weight $\lambda=1e-1$ versus a large value $\lambda=1e+4$. The results are shown in Figure~\ref{fig1} for tracking a sinusoidal reference with the frequency of 1Hz, duration of 4 seconds and 100 samples per second, i.e. $t=0:0.01:4$, $r=\sin(2\pi F t)$, which is used in all simulations.      
\begin{figure}[h]
\centering
    \includegraphics[width=1\columnwidth]{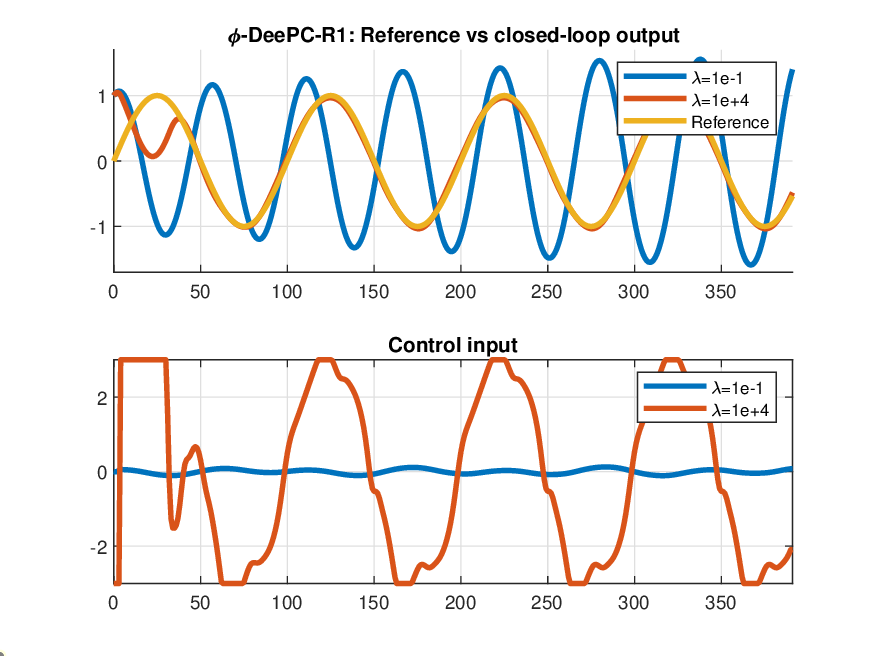}
 \caption{Tracking performance of $\phi$-DeePC-R1 for noise free data and various $\lambda$.}
\label{fig1}
\end{figure}

We observe that for the small value of $\lambda$, the predictions of $\phi$-DeePC are not consistent as expected, and the behavior is similar to that of un-regularized linear DeePC in the presence of noisy data, i.e., the resulting input is close to zero and the output is far from the reference, while the predictive controller estimates that such an input should give good results. However, once the variable $\lambda$ is sufficiently large, we obtain consistent predictions and good tracking performance, as indicated by Lemma~\ref{lem:eq}.

\begin{table}[]
\centering
\begin{tabular}{|l|l|l|l|l|l|l|}
\hline
\textbf{Formulation} & $J_{ISE}$ & $J_{IAE}$  & $J_u$     & $J_{track}$  & CPU\\ \hline
$\phi$-SPC      &    0.0428    &    0.0840     &    2.0796  &    11.1387 & 0.0087\\ 
\hline
$\phi$-DeePC-R1       &  0.0429
     &     0.0841     & 2.0758 & 11.1373 & 0.0955          \\ \hline
$\phi$-DeePC-R2      &    0.0429   &     0.0841       &  2.0758 &  11.1373 & 0.0806       \\ \hline
Ridge $\phi$-DeePC     &0.0429
    & 0.0850        & 2.0773 & 11.1469 & 0.0083         \\ \hline
\end{tabular}
\caption{Performance 
\& mean CPU time for noiseless data.}
\label{Tab:nonoise}
\end{table}

To compare the performance and computational complexity of all the derived data-driven predictive controllers in the noise free case, we report the following performance indexes in Table~\ref{Tab:nonoise}: the integral squared error $J_{ISE}=\frac{1}{T_{sim}}\sum_{k=1}^{T_{\text{sim}}}\|y(k)-r(k)\|_2^2$, the integral absolute error $J_{IAE}=\frac{1}{T_{sim}}\sum_{k=1}^{T_{\text{sim}}}\|y(k)-r(k)\|_1$, the input cost $ J_u=\frac{1}{T_{sim}}\sum_{k=1}^{T_{\text{sim}}}\|u(k)\|_1$ and the tracking cost $J_{track}=\frac{1}{T_{sim}}(\sum_{k=1}^{T_{sim}}\|Q^\frac{1}{2}(y(k)-y_r(k))\|_2^2+\|R^\frac{1}{2}(u(k)-u_r(k))\|_2^2)$. The mean CPU time in seconds is also given in Table~\ref{Tab:nonoise}. 

We observe that as expected the $\phi$-SPC and Ridge $\phi$-DeePC formulations are computationally much more efficient, while the $\phi$-DeePC-R1 formulation yields equivalent performance with $\phi$-DeePC-R2 and both yield slightly better tracking performance overall. This is consistent with the behavior of linear DeePC, which by optimizing the variance-bias trade-off in the noisy data case, can obtain better performance than linear SPC. In this case, although there is no measurement noise, since the residuals of the $\phi$-SPC predictor are not equal to zero, there is a prediction error and $\phi$-DeePC-R1(-R2) can better compensate for it.

\begin{figure}[h]
\centering
    \includegraphics[width=1\linewidth]{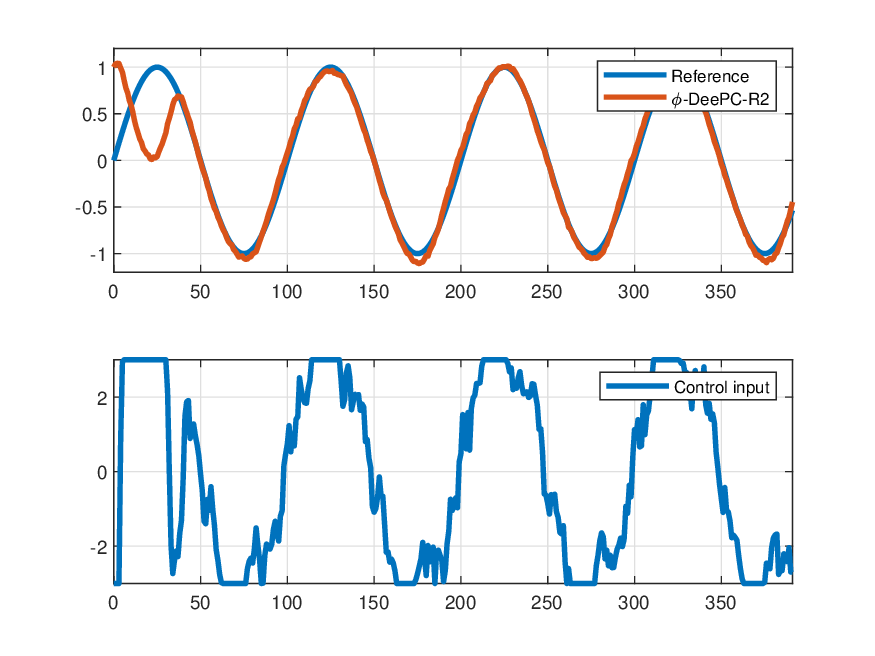}
 \caption{$\phi$-DeePC-R2: tracking performance for noisy data.}
\label{fig2}
\end{figure}

Next, we illustrate the tracking performance for noisy data of the $\phi$-DeePC-R2 in Figure~\ref{fig2}. As guaranteed by Lemma~\ref{lem:eq}, $\phi$-DeePC-R2 is robust to noisy data and results in very good tracking performance, despite a non-negligible measurement noise standard deviation. The resulting performance indicators are reported for all data-driven predictive controllers in Table~\ref{Tab:noise} along with mean CPU times. 
\begin{table}[]
\centering
\begin{tabular}{|l|l|l|l|l|l|l|}
\hline
\textbf{Formulation} & $J_{ISE}$ & $J_{IAE}$  & $J_u$     & $J_{track}$  & CPU\\ \hline
$\phi$-SPC      &    0.0474   &    0.0979     &    2.2209
  &    12.2987 & 0.0084\\ 
\hline
$\phi$-DeePC-R1       &  0.0466
     &    0.0932     & 2.1642 & 12.0210 & 0.1013          \\ \hline
$\phi$-DeePC-R2      &    0.0466  &     0.0932
       &  2.1642&  12.0210 & 0.0852       \\ \hline
Ridge $\phi$-DeePC     & 0.0467
    & 0.0982
        & 2.1815 & 12.0709 & 0.0090         \\ \hline
\end{tabular}
\caption{Performance 
\& mean CPU time for noisy data.}
\label{Tab:noise}
\end{table}
In terms of computational complexity $\phi$-SPC and Ridge $\phi$-DeePC remain much more efficient, as expected. In terms of tracking performance, in the noisy data case, the ability of $\phi$-DeePC-R1(-R2) to optimize the bias/variance trade-off plays out even more than in the noise free case, resulting in relatively better tracking performance compared to $\phi$-SPC, versus the noise free data case.
\section{Conclusions}
\label{sec5}
We have provided a general basis functions formulation of nonlinear data-enabled predictive control. We have presented necessary and sufficient conditions for consistency of basis functions behavioral multi-step predictors in relation with basis functions identified multi-step predictors. From these conditions we derived two regularized formulations of basis functions DeePC with guaranteed consistent predictions for both noise free and noisy data. To reduce computational complexity we have also derived a basis function DeePC formulation based on ridge regression, which has significantly less optimization variables. All the derived formulations were validated on a nonlinear pendulum state-space model from the literature, where they showed consistenly superior tracking performance, also in the presence of noisy data. The consistency result in Lemma~\ref{lem:eq} opens the door to utilizing a wide range of powerful machine learning methods \cite{Suykens} for data-enabled predictive control of nonlinear systems, which is very appealing for many real-life applications.
\paragraph*{Acknowledgements}
The author gratefully acknowledges the assistance provided by MSc. Mihai-Serban Popescu with constructing and training the hybrid radial basis functions neural networks in PyTorch required to generate the basis functions representation and with preparing the $\phi$-SPC Matlab simulation files. 

\end{document}